\documentclass[11pt,a4paper,twoside]{amsart}
\usepackage[activeacute,english]{babel}
\usepackage[latin1]{inputenc}
\usepackage{amsfonts}
\usepackage{amsmath,amsthm}
\usepackage{amssymb}
\usepackage{graphics}
\usepackage{yfonts}
\usepackage{upgreek}
\usepackage{epstopdf}

\newtheorem{theorem}{Theorem}[section]
\newtheorem{lemma}[theorem]{Lemma}
\newtheorem{proposition}[theorem]{Proposition}
\newtheorem{corollary}[theorem]{Corollary}
{\theoremstyle{remark} \newtheorem{remark}[theorem]{Remark}}

\newcommand{\CC}{{\mathcal C}}

\newcommand{\VV}{{\mathcal V}}
\newcommand{\XX}{{\mathcal X}}

\newcommand{\C}{{\mathbb C}}

\newcommand{\R}{{\mathbb R}}

\newcommand{\Tr}{{\operatorname{t}}}
\newcommand{\Ran}{{\operatorname{Range}}}
\newcommand{\Ker}{{\operatorname{Ker}}}

\title[Dirac operators, shell interactions and gauge functions]{Dirac operators, shell interactions and\\discontinuous gauge functions across the boundary}
\author[A. Mas]{Albert Mas}
\date{}
\thanks{Mas was supported by the {\em Juan de la Cierva} program JCI2012-14073 and the spanish funding projects MTM2011-27739 and MTM2014-52402.}
\address{A. Mas. Departament de Matem\`atica Aplicada I,
ETSEIB, Universitat Polit\`ecnica de Catalunya, Avda. Diagonal 647, 08028 Barcelona (Spain)}
\email{amasblesa@gmail.com}

\begin{document}

\begin{abstract}
Given a bounded smooth domain $\Omega\subset\R^3$, we explore the relation between couplings of the free Dirac operator $-i\alpha\cdot\nabla+m\beta$ with pure electrostatic shell potentials $\lambda\delta_{\partial\Omega}$ ($\lambda\in\R$) and some perturbations of those potentials given by the normal vector field $N$ on the shell $\partial\Omega$, namely $\{\lambda_e+\lambda_n(\alpha\cdot N)\}\delta_{\partial\Omega}$ ($\lambda_e$, $\lambda_n\in\R$). Under the appropiate change of parameters, the couplings with perturbed and unperturbed electrostatic shell potentials yield unitary equivalent self-adjoint operators. The proof relies on the construction of an explicit family of unitary operators that is well adapted to the study of shell interactions, and fits within the framework of gauge theory. A generalization of such unitary operators also allow us to deal with the self-adjointness of couplings of 
$-i\alpha\cdot\nabla+m\beta$ with some shell potentials of magnetic type, namely $\lambda(\alpha\cdot N)\delta_{\partial\Omega}$ with $\lambda\in\CC^1(\partial\Omega)$.
\end{abstract}

\maketitle

\section{Introduction}
The main purpose of this paper is to explore the relation between couplings of the free Dirac operator with pure electrostatic shell potentials and some concrete perturbations of those potentials given by the normal vector field on the shell where the formers are defined. The main result in this article states that, under the appropiate change of parameters, the couplings with perturbed and unperturbed electrostatic shell potentials yield unitary equivalent self-ajoint operators. This is proven by constructing an explicit family of simple unitary operators that relates both couplings. A generalization of such unitary operators also allow us to deal with the self-adjointness of couplings of the free Dirac operator with some regular shell potentials of magnetic type.

The free Dirac operator in $\R^3$ is defined by
$H=-i\alpha\cdot\nabla+m\beta,$
where $\alpha=(\alpha_1,\alpha_2,\alpha_3)$,
\begin{equation}\label{freedirac}
\begin{split}
\alpha_j
=&\left(\begin{array}{cc} 0 & \sigma_j\\
\sigma_j & 0 \end{array}\right)
\quad\text{for }j=1,2,3,\quad
\beta
=\left(\begin{array}{cc} I_2 & 0\\
0 & -I_2 \end{array}\right),\quad
I_2
=\left(\begin{array}{cc} 1 & 0\\
0 & 1 \end{array}\right),\\
&\text{and}\quad
\sigma_1
=\left(\begin{array}{cc} 0 & 1\\
1 & 0 \end{array}\right),
\quad\sigma_2
=\left(\begin{array}{cc} 0 & -i\\
i & 0 \end{array}\right),
\quad\sigma_3
=\left(\begin{array}{cc} 1 & 0\\
0 & -1 \end{array}\right)
\end{split}
\end{equation}
compose the family of {\em Pauli matrices}. Note that $H$ acts on functions $\varphi:\R^3\to\C^4$. Throughout this article we assume $m>0$.

The shell where the potentials are defined corresponds to the boundary of a bounded smooth domain $\Omega\subset\R^3$. Let $\sigma$ and $N$ be the surface measure and outward unit normal vector field on ${\partial\Omega}$, respectively. For convenience of notation, we set $\Omega_+=\Omega$ and $\Omega_-=\R^3\setminus\overline{\Omega}$, so ${\partial\Omega}=\partial\Omega_\pm$. Given $\lambda\in\R$ and $\varphi:\R^3\to\C^4$, the electrostatic shell potential $V_\lambda$ applied to $\varphi$ is formally defined as
$$V_\lambda\varphi=\lambda\frac{\varphi_++\varphi_-}{2}\sigma,$$
where $\varphi_\pm$ denote the boundary values of $\varphi$ (whenever they exist in a reasonable sense) when one approaches to ${\partial\Omega}$ from $\Omega_\pm$. Therefore, $V_\lambda$ maps functions defined in $\R^3$ to vector measures of the form $f\sigma$ with $f:{\partial\Omega}\to\C^4$. In particular, one can interpret $V_\lambda\varphi$ as the distribution $\lambda\varphi\delta_{\partial\Omega}$ when acting on functions $\varphi$ which have a well-defined trace on ${\partial\Omega}$, where $\delta_{{\partial\Omega}}$ denotes the Dirac-delta distribution on ${\partial\Omega}$. The self-adjoint character of $H+V_\lambda$ was already treated in \cite{AMV1}, although previous results in the case that $\partial\Omega$ is a sphere or, much more in general, viewing $H+V_\lambda$ as a particular instance of singular perturbations of self-adjoint operators were obtained in \cite{Dittrich} and \cite{Posi1}, respectively. 

The perturbed electrostatic shell potentials mentioned before are given, for $\lambda_e,\lambda_n\in\R$, by shell potentials of the type
\begin{equation*}
V_{\lambda_e,\lambda_n}\varphi=\{\lambda_e+\lambda_n(\alpha\cdot N)\}\frac{\varphi_++\varphi_-}{2}\sigma.
\end{equation*}
From the previous definitions, $V_\lambda=V_{\lambda,0}$. The self-adjoint character of $H+V_{\lambda_e,\lambda_n}$ can be also dealt with the results in \cite{Posi1} but, for the reader's convenience, in this article we include the construction of this self-adjoint operator.

The main purpose of this note is to show that, in general, the spectral study of $H+V_{\lambda_e,\lambda_n}$ can be reduced to the one of $H+V_\lambda$ for some properly chosen $\lambda\in\R$ related to $\lambda_e$ and $\lambda_n$. This reduction relies on the fact that there exists a unitary equivalence between these two operators. Moreover, we show the explicit connection between $\lambda_e$, $\lambda_n$ and $\lambda$ which, indeed, is independent of the underlying domain $\Omega$.  
As a consequence of this unitary equivalence, all the results about existence of pure point spectrum in $(-m,m)$ for $H+V_\lambda$ obtained in \cite{AMV2} as well as the isoperimetric-type inequality for electrostatic shell potentials shown in \cite{AMV3} can be transferred to the perturbed case $H+V_{\lambda_e,\lambda_n}$, once they are properly restated in terms of $\lambda_e$ and $\lambda_n$. As a byproduct, the isospectral transformation 
\begin{equation}\label{xxx}
H+V_\lambda\longleftrightarrow H+V_{-4/\lambda}
\end{equation}
obtained in \cite{AMV2} shows up as a particular example of the construction developed here.

Most of the results in this article fit within the framework of {\em gauge transformations}, but where the gauge functions are discontinuous across $\partial\Omega$. More precisely, since the function $\eta=\lambda_n\chi_{\Omega_-}$ is constant in $\Omega_\pm$ and satisfies $\eta_-=\lambda_n$ and $\eta_+=0$ on $\partial\Omega$, where $\eta_\pm$ denote the boundary values of $\eta$ when we approach to $\partial\Omega$ from $\Omega_\pm$, then $$\nabla\eta=\lambda_n N\sigma\quad\text{and}\quad
\alpha\cdot\nabla\eta=\lambda_n(\alpha\cdot N)\sigma.$$ 
But $\nabla\times(\lambda_n N\sigma)
=\nabla\times\nabla\eta=0$ in the sense of distributions, and thus $\lambda_n(\alpha\cdot N)\sigma$ corresponds to a (discontinuous across the boundary) change of gauge in the Dirac operator. However, at the end of the article, we also deal with the self-adjoint character of shell interactions of the form
\begin{equation}\label{papapa}
(H+\VV_{\lambda})\varphi=H\varphi+\lambda(\alpha\cdot N)\frac{\varphi_++\varphi_-}{2}\sigma
\end{equation}
for $\lambda\in\CC^1(\partial\Omega)$, which can be interpreted as a magnetic shell interaction in the normal direction in case that $\lambda$ is non-constant (we wrote $\VV_\lambda$ instead of $V_{0,\lambda}$ to make distinction between the non-constant and constant case, respectively).

Regarding the structure and the concrete contents of the article, in Section \ref{s preli} we recall some preliminary facts, all of them extracted from \cite{AMV1}, which deal with a general construction of self-adjoint shell interactions for Dirac operators (see Therorem \ref{pre t1}) as well as the boundary behaviour on $\partial\Omega$ of the functions in the domain of definition of those operators (see Lemma \ref{l jump}). In Section \ref{s uni equiv} we introduce a family of unitary operators that, when applied to the shell interactions presented in Section \ref{s preli}, allow us to generate a collection of unitary equivalent operators whose description in the terms of Section \ref{s preli} is also given (see Lemma \ref{lema elec}).
The construction of these unitary maps is based on the one developed in the proof of \cite[Theorem 3.6]{BEL}. 

Section \ref{s applications} contains the main applications of the abstact results developed in the previous sections. First, in Theorem \ref{t elec1} we provide an explicit description of the domain of definition where $H+V_{\lambda_e,\lambda_n}$ is self-adjoint. Then, in Theorem \ref{main apli theo} we present all the possible unitary equivalent self-adjoint operators $H+V_{\lambda_e',\lambda_n'}$ that we can obtain (with the method developed in Section \ref{s uni equiv}) from a given self-adjoint one $H+V_{\lambda_e,\lambda_n}$ and which can be described in the same terms as the former. In the subsequent corollaries we explore the scope of Theorem \ref{main apli theo}. For the reader's convenience, those corollaries are stated below somewhat informally (in particular, we omit the description of the domains where the couplings are defined, which are detailed in Section \ref{s applications}). Theorems \ref{coro1 inf}, \ref{coro2 inf} and \ref{coro3 inf} correspond to Corollaries \ref{coro1}, \ref{coro2} and \ref{coro3}, respectively.
\begin{theorem}\label{coro1 inf}
Let $\lambda_e,\lambda_n\in\R$ be such that $\lambda_e^2-\lambda_n^2\neq 0,4$. Then
$$H+V_{\lambda_e,\lambda_n}\text{ and }
H+V_{-4\lambda_e/(\lambda_e^2-\lambda_n^2),
4\lambda_n/(\lambda_e^2-\lambda_n^2)}$$
are unitary equivalent self-adjoint operators. 
\end{theorem} 
If we choose $\lambda_n=0$ in Theorem \ref{coro1 inf}, we get the isospectral relation \eqref{xxx}.

\begin{theorem}\label{coro2 inf}
Let $\lambda_e,\lambda_n\in\R$ be such that $\lambda_e\neq0$ and $\lambda_e^2-\lambda_n^2=-4$. 
Then 
$$H+V_{\lambda_e,\lambda_n}\text{ and }
H+V_{(\pm2\lambda_n-4)/\lambda_e,0}$$
are unitary equivalent self-adjoint operators. 
\end{theorem}
This theorem may be seen as a particular instance of the following one (if we do not pay attention to the assumptions on $\lambda_e$ and $\lambda_n$, it would correspond to $\theta=\pm\pi/2$), but we stated it separately because its proof and conclusion are simpler.

\begin{theorem}\label{coro3 inf}
Let $\lambda_e,\lambda_n\in\R\setminus\{0\}$ be such that 
$|\lambda_e^2-\lambda_n^2|\neq0,4$,
\begin{equation}\label{extra assumpt}
\lambda_e^2-\lambda_n^2+2\lambda_e\neq4
\text{ and }\lambda_e^2-\lambda_n^2-2\lambda_e\neq4.
\end{equation}
Then there exists $\theta\in\R$ such that
$$H+V_{\lambda_e,\lambda_n}\text{ and }
H+V_{\left(2\lambda_n\frac{1+\cos\theta}{\sin\theta}-4\right)/\lambda_e,0}$$
are unitary equivalent self-adjoint operators.
\end{theorem}

Concerning the assumptions in the theorems above, let us mention that 
$\lambda_e^2-\lambda_n^2\neq0,4$ is the only requirement that we need to show that $H+V_{\lambda_e,\lambda_n}$ is self-adjoint and to construct a unitary equivalent operator which can be understood as a coupling of $H$ with a pure electrostatic shell potential $V_{\lambda,0}$ for some $\lambda\in\R$ given by Theorems \ref{coro2 inf} and \ref{coro3 inf}. The extra assumption \eqref{extra assumpt} is only used to describe this new operator in the same vein as $H+V_{\lambda_e,\lambda_n}$ (and not by simply using the unitary transformations, that is, by defining it just as $U^{-1}(H+V_{\lambda_e,\lambda_n})U$ for some unitary map $U$).
In Figure \ref{fig1}, the black curves correspond to the points 
$(\lambda_e,\lambda_n)\in\R^2$ such that $\lambda_e^2-\lambda_n^2=0$ or $\lambda_e^2-\lambda_n^2=4$ and the red curves represent the points such that \eqref{extra assumpt} does not hold. 

\begin{figure}[ht]
\begin{center}
\scalebox{0.45}{\includegraphics{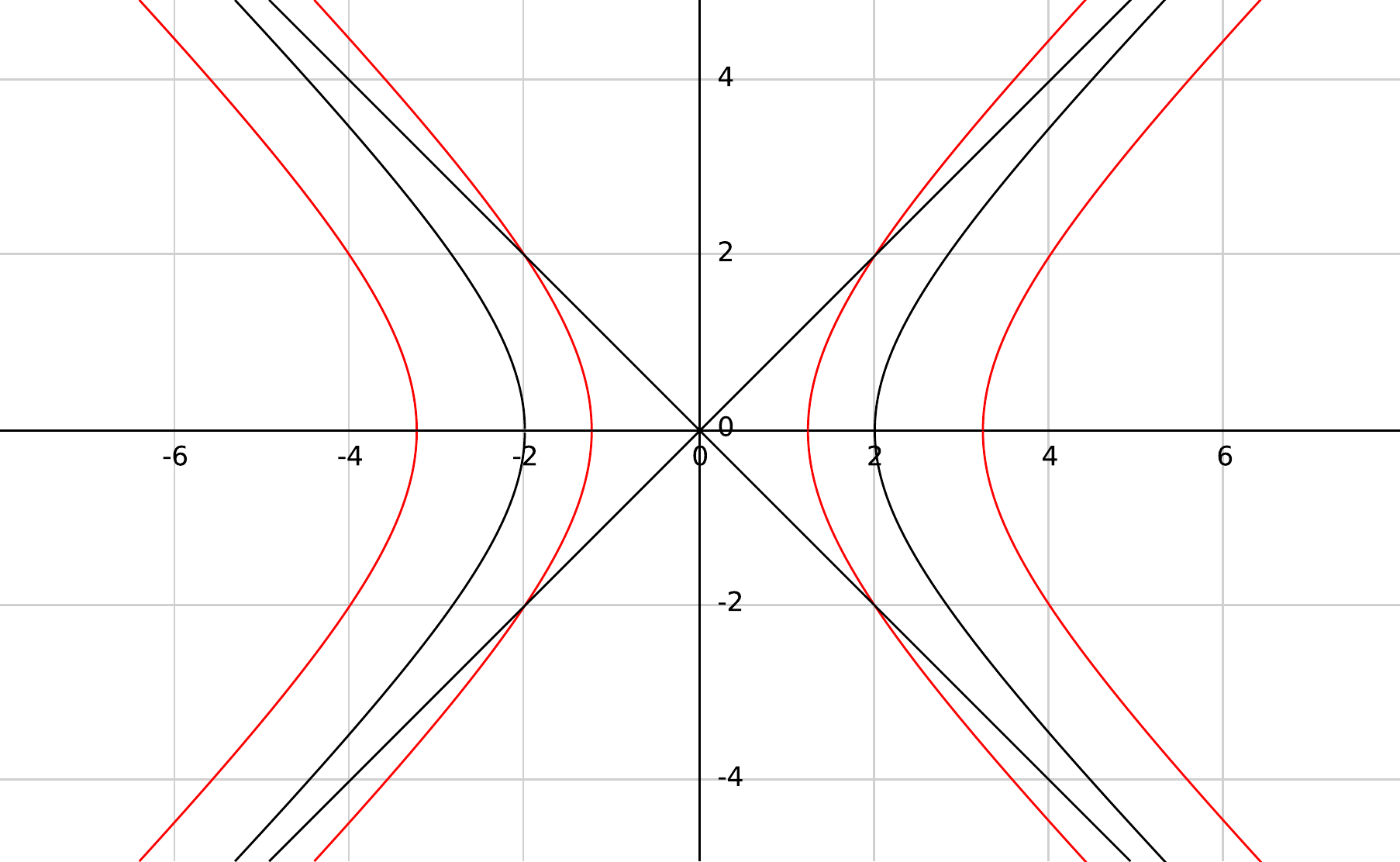}}
\caption{\label{fig1}Picture of some of the assumptions on $(\lambda_e,\lambda_n)\in\R^2$.}
\end{center}
\end{figure}

A final remark regarding the above-mentioned results is in order. Since we are mainly interested on couplings with non-trivial electrostatic shell potentials, some of the results in this paper assume that $\lambda_e\neq0$. The case of unitary equivalence between the free Dirac operator and its shell interactions with potentials of the type $V_{0,\lambda_n}$ must be treated separately. However, similar arguments to the ones in this article and in \cite{AMV1} work in that case (for instance, one should use \cite[Theorem 2.11$(i)$]{AMV1} with $\Lambda\equiv0$ instead of \cite[Theorem 2.11$(iii)$]{AMV1} to define $H+V_0\equiv H$ in Theorem \ref{pre t1} below). For shortness sake, we don't carry on the study of this particular case in this note.

Finally, in Section \ref{s nonc magnetic} we give an explicit construction of a self-adjoint operator that realizes the coupling given by \eqref{papapa}, first in the case that $\lambda\in\CC^1(\partial\Omega)$ does not vanish on $\partial\Omega$ (see Proposition \ref{propo magnetic}), and then in the general case (see Theorem \ref{t nonc magnetic}).

\section*{Acknowledgment}
The author gratefully acknowledges Eric S\'er\'e for useful conversations while preparing this article during a short stay at Universit\'e Paris-Dauphine, as well as the facilities provided by the center.

\section{Preliminaries}\label{s preli}
This article relies on \cite{AMV1,AMV2}, so we assume that the reader is familiar with the notation, methods and results in there. However, in this section we recall some basics of those articles for the reader's convenience.

A fundamental solution of $H$ is given by
$$\phi(x)=\frac{e^{-m|x|}}{4\pi|x|}\left(m\beta
+\left(1+m|x|\right)i\alpha\cdot\frac{x}{|x|^2}\right)\quad\text{for }x\in\R^3\setminus\{0\},$$
i.e., $H\phi=\delta_0$ in the sense of distributions, where $\delta_0$ denotes the Dirac measure centered at the origin. This fundamental solution is the first key ingredient for the developements below.

We denote by $\mu$ the Lebesgue measure in $\R^3$ and by $\sigma$ the surface measure on $\partial\Omega$. 
We define the auxiliar space of locally finite $\C^4$-valued Borel measures $$\XX=\left\{G\mu+g\sigma:\,G\in L^2(\mu)^4,\, g\in L^2(\sigma)^4\right\},$$
where, as usual, 
$L^2(\nu)^4=\left\{f:\R^3\to\C^4\text{ $\nu$-measurable}:\,
\int|f|^2\,d\nu<\infty\right\}$ for a given positive Borel measure $\nu$ in $\R^3$.
In what follows we use a nonstandard notation, $\Phi$, to define the convolution of measures in $\XX$ with the fundamental solution of $H$, $\phi$. Capital letters, as $F$ or $G$, in the argument of $\Phi$ denote elements of $L^2(\mu)^4$, and the lowercase letters, as $f$ or $g$, denote elements in $L^2(\sigma)^4$. 
Despite that this notation is nonstandard, it is very convenient in order to shorten the forthcoming computations.

Given $G\mu+g\sigma\in\XX$, we define $$\Phi(G+g)=
\phi*G\mu+\phi*g\sigma\in L^2(\mu)^4,$$
where, for a given locally finite $\C^4$-valued Borel measure $\nu$ in $\R^3$ and $x\in\R^3$, we have set $\phi*\nu(x)=\int\phi(x-y)\,d\nu(y)$ whenever the integral makes sense.
One can check that
$H(\Phi(G+g))=G\mu+g\sigma$ in the sense of distributions. This allows us to define a ``generic'' potential $V$ acting on functions $\varphi=\Phi(G+g)$ for $G\mu+g\sigma\in\XX$ by
$V(\varphi)= -g\sigma$, so that $(H+V)(\varphi)=G\mu$
in the sense of distributions. Recall that $G\in L^2(\mu)^4$ and $g\in L^2(\sigma)^4$ so, under this setting, we make the following abuse of notation:
\begin{equation}\label{defi V}
\begin{split}
&\text{For } \varphi=\Phi(G+g) \text{ with } G+g\in\XX,
\text{ set }V(\varphi)= -g\in L^2(\sigma)^4.\\ 
&\text{Then }H\varphi=G+g\in\XX\text{ and }(H+V)\varphi=G\in L^2(\mu)^4. 
\end{split}
\end{equation}
That is, we omit the underlying measures
$\mu$ and $\sigma$ associated to $G$ and $g$, respectively, in the distributional relations stated in \eqref{defi V}. In particular, we have $\Phi(\XX)\subset L^2(\mu)^4$ and $H+V:\Phi(\XX)\to L^2(\mu)^4$.
This abuse of notation will be systematically used throughout the article.

We also make use of the trace operator as follows. For $G\in\CC_c^\infty(\R^3)^4$, one defines the trace on ${\partial\Omega}$ by $\Tr_{\partial\Omega}(G)=G\chi_{{\partial\Omega}}$. Then, $\Tr_{\partial\Omega}$ extends to a bounded linear operator $\Tr_\sigma:W^{1,2}(\mu)^4\to L^2(\sigma)^4$, where $W^{1,2}(\mu)^4$ denotes the Sobolev space of $\C^4$-valued functions such that all its components have all its derivatives up to first order in $L^2(\mu)$. Since $\Phi(G)\in W^{1,2}(\mu)^4$ for all $G\in L^2(\mu)^4$, we can define $$\Phi_\sigma G=\Tr_\sigma(\Phi(G))\in L^2(\sigma)^4.$$
That $\Phi$, $V$ and $\Phi_\sigma$ are well defined and satisfy the above-mentioned properties is justified in \cite[Section 2.3]{AMV1}.

From the comments above and following \cite{AMV1}, we are ready to construct  domains where $H+V$ is self-adjoint. This is done with the help of auxiliar operators. More precisely, given a linear operator $\Lambda$ bounded in $L^2(\sigma)^4$, set 
\begin{equation}\label{def T_Lambda}
\begin{split}
&D(T_\Lambda)=\{\Phi(G+g): G+g\in\XX,\,\Phi_\sigma G =\Lambda g\}\subset L^2(\mu)^4,\\
&T_\Lambda=H+V:D(T_\Lambda)\to L^2(\mu)^4, \text{ where }V \text{ is as in }\eqref{defi V}.
\end{split}
\end{equation} 
The following theorem, which is a direct application of 
\cite[Theorem 2.11$(iii)$]{AMV1}, shows that $T_\Lambda$ given by \eqref{def T_Lambda} is self-adjoint under some assumptions on $\Lambda$.  
\begin{theorem}\label{pre t1}
Let $T_\Lambda$ be as in \eqref{def T_Lambda}. If $\Lambda$ is self-adjoint, $\Ran(\Lambda)$ is closed in $L^2(\sigma)^4$ and $\{{\Phi(g)}:\,g\in\Ker(\Lambda)\}$ is closed in $L^2(\mu)^4$, then $T_\Lambda$ is  self-adjoint. In particular, if $\Lambda$ is self-adjoint and Fredholm, then $T_\Lambda$ is self-adjoint.
\end{theorem}

The next lemma describes the traces on $\partial\Omega$ of functions $\varphi\in\Phi(\XX)$, and it will be used in the sequel (see \cite[Lemma 3.3]{AMV1} for a proof).
\begin{lemma}\label{l jump}
Given $\varphi=\Phi(G+g)$ with $G+g\in \XX$ and $x\in{\partial\Omega}$, set 
\begin{equation*}
\begin{split}
\varphi_\pm(x)&=\Phi_\sigma G(x)+\textstyle{\lim_{\Omega_{\pm}\ni y\stackrel{nt}{\longrightarrow} x}\Phi(g)(y)}
\quad\text{and}\\
C_\sigma g(x)
&=\textstyle{\lim_{\epsilon\searrow0}\int_{|x-z|>\epsilon}\phi(x-z)g(z)\,d\sigma(z),}
\end{split}
\end{equation*}
where $\textstyle{\Omega_{\pm}\ni y\stackrel{nt}{\longrightarrow} x}$ means that $y\in\Omega_{\pm}$ tends to $x\in{\partial\Omega}$ non-tangentially. Then, $\varphi_\pm$ are well defined $\sigma$-a.e. \!on $\partial\Omega$, $C_\sigma$ is bounded and self-adjoint in $L^2(\sigma)^4$ and, furthermore,
\begin{itemize}
\item[$(i)$] $\varphi_\pm =\Phi_\sigma G+\left(\mp\frac{i}{2}\,(\alpha\cdot N)+C_\sigma\right)\!g\in L^2(\sigma)^4$,
\item[$(ii)$] $(C_\sigma(\alpha\cdot N))^2=-\frac{1}{4}$.
\end{itemize}
\end{lemma}
Finally, recall that $(\alpha\cdot x)(\alpha\cdot x)=|x|^2$ for all $x\in\R^3$, thus $(\alpha\cdot N)^2=1$. Besides, the algebraic properties of the $\alpha_j$'s and $\beta$ also yield $\beta(\alpha\cdot N)=(\alpha\cdot N)\beta=0$. These are two basic facts that will also be used in what follows.

\section{Unitary equivalence}\label{s uni equiv} 
Given $T_\Lambda$ as in \eqref{def T_Lambda}, we are going to construct a family of unitary equivalent operators parametrized by the complex numbers with unit modulus. Given $z\in\C$ with $|z|=1$, define the operator $U_z:L^2(\mu)^4\to L^2(\mu)^4$ by $$U_z\varphi=(\chi_{\Omega_+}
+\overline{z}\chi_{\Omega_-})\varphi,$$
where $\chi_{\Omega_\pm}$ denotes the characteristic function of $\Omega_\pm$.
We see that $(U_z)^*=U_{\overline z}$ and $U_zU_{\overline{z}}=U_{\overline{z}}U_z=1$, so $U$ is unitary.
Given $T_\Lambda$ as in \eqref{def T_Lambda}, set 
\begin{equation}\label{uni equiv def}
(T_\Lambda)_z=U_{\overline z}T_\Lambda U_z\quad
\text{defined on}\quad D((T_\Lambda)_z)=U_{\overline z}D(T_\Lambda).
\end{equation} 
Then $T_\Lambda$ and $(T_\Lambda)_z$ are unitary equivalent operators. For the applications below, we want to find a description of $(T_\Lambda)_z$ and $D((T_\Lambda)_z)$ similar to the one of $T_\Lambda$ and $D(T_\Lambda)$ in \eqref{def T_Lambda}. This is precisely the purpose of the following lemma.

\begin{lemma}\label{lema elec}
Let $T_\Lambda$ be as in \eqref{def T_Lambda} and, for $z\in\C$ with $|z|=1$, let $(T_\Lambda)_z$ be the unitary equivalent operator given by \eqref{uni equiv def}. If there exists a linear operator $\Lambda_z$ bounded in $L^2(\sigma)^4$ such that the pair $(\Lambda,\Lambda_z)$ satisfies
\begin{equation}\label{uni equiv eq1}
\Lambda_z\bigg(\frac{1+z}{2}+(1-z)i(\alpha\cdot N)(\Lambda+C_\sigma)\bigg)=\bigg(\frac{1+z}{2}-(1-z)iC_\sigma(\alpha\cdot N)\bigg)\Lambda
\end{equation}
for $C_\sigma$ as in {\em Lemma \ref{l jump}},
then $(T_\Lambda)_z\subset T_{\Lambda_z}$, where $T_{\Lambda_z}$ is defined by \eqref{def T_Lambda}. 
\end{lemma}
\begin{proof}
Let $\varphi=\Phi(G+g)\in D(T_\Lambda)$. Reasoning as in the proof of \cite[Lemma 5.1]{AMV2} and, more precisely, applying $\Phi$ to \cite[equation (5.4)]{AMV2}, we see that $$\chi_{\Omega_\pm}\varphi
=\Phi\bigg(\chi_{\Omega_\pm}G+\bigg(\frac{1}{2}\pm i(\alpha\cdot N)(\Lambda+C_\sigma)\bigg)g\bigg),$$ which, from the definition of $U_{\overline z}$, easily implies that 
\begin{equation}\label{uu eq1}
U_{\overline z}\varphi
=\Phi\bigg(U_{\overline z}G+\bigg(\frac{1+z}{2}+(1-z)i(\alpha\cdot N)(\Lambda+C_\sigma)\bigg)g\bigg),
\end{equation} 
and then, using \eqref{defi V}, $(H+V)U_{\overline z}\varphi=U_{\overline z}G.$
Moreover, by \eqref{uni equiv def} we see that $U_{\overline z}\varphi\in D((T_\Lambda)_z)$ and, in view of \eqref{def T_Lambda}, \eqref{defi V} and \eqref{uu eq1}, we get
\begin{equation}\label{uu eq2}
(T_\Lambda)_zU_{\overline z}\varphi=U_{\overline z}T_\Lambda U_zU_{\overline z}\varphi=U_{\overline z}T_\Lambda\varphi=
U_{\overline z}(H+V)\varphi=U_{\overline z}G=(H+V)U_{\overline z}\varphi.
\end{equation}

Combining \cite[equation (5.3)]{AMV2} with the fact that $\Phi_\sigma G=\Lambda g$ by \eqref{def T_Lambda}, we deduce that $$\Phi_\sigma(\chi_{\Omega_\pm}G)=\bigg(\frac{1}{2}\mp iC_\sigma (\alpha\cdot N)\bigg)\Lambda g,$$
and hence 
\begin{equation}\label{uu eq3}
\Phi_\sigma(U_{\overline z}G)=\Phi_\sigma(\chi_{\Omega_+}G)
+z\Phi_\sigma(\chi_{\Omega_-}G)
=\bigg(\frac{1+z}{2}-(1-z)iC_\sigma(\alpha\cdot N)\bigg)\Lambda g.
\end{equation} 

Finally, assume that \eqref{uni equiv eq1} holds for some $\Lambda_z$ bounded in $L^2(\sigma)^4$.
Then, by setting $$F=U_{\overline z}G\in L^2(\mu)^4\quad
\text{and}\quad f=\bigg(\frac{1+z}{2}+(1-z)i(\alpha\cdot N)(\Lambda+C_\sigma)\bigg)g\in L^2(\sigma)^4,$$
a combination of \eqref{uu eq1}, \eqref{uu eq3} and \eqref{uni equiv eq1} shows that
$$D((T_\Lambda)_z)\subset\big\{\Phi(F+f): F+f\in\XX,\,\Phi_\sigma F =\Lambda_z f\big\},$$ and \eqref{uu eq2} gives $(T_\Lambda)_z=H+V$ on $D((T_\Lambda)_z).$ Using \eqref{def T_Lambda}, these last conclusions mean that $(T_\Lambda)_z\subset T_{\Lambda_z}$, and the lemma is proved. 
\end{proof}

\section{Electrostatic shell potentials and gauge transformations}\label{s applications}
In this section we study the unitary equivalence developed in Section \ref{s uni equiv} applied to couplings of $H$ with the perturbed and unperturbed electrostatic shell potentials associated to $\partial\Omega$ which were presented in the introduction. As a first step, we deal with the self-adjoint character of such concrete couplings via the following theorem.
\begin{theorem}\label{t elec1}
Given $\lambda_e,\lambda_n\in\R$ such that $\lambda_e^2-\lambda_n^2\neq0,4$ and $C_\sigma$ as in {\em Lemma \ref{l jump}}, set
\begin{equation}\label{pre eq1}
\Lambda_{\lambda_e,\lambda_n}=\frac{\lambda_n(\alpha\cdot N)-\lambda_e}{\lambda_e^2-\lambda_n^2}-C_\sigma.
\end{equation}
Then $T_{\Lambda_{\lambda_e,\lambda_n}}$ given by \eqref{def T_Lambda} is self-adjoint and $T_{\Lambda_{\lambda_e,\lambda_n}}=H+V_{\lambda_e,\lambda_n}$ on $D(T_{\Lambda_{\lambda_e,\lambda_n}})$, where $V_{\lambda_e,\lambda_n}$ is defined, for $\varphi\in \Phi(\XX)$ and $\varphi_\pm$ as in {\em Lemma \ref{l jump}}, by
\begin{equation}\label{pre eq2}
V_{\lambda_e,\lambda_n}\varphi=\{\lambda_e+\lambda_n(\alpha\cdot N)\}\frac{\varphi_++\varphi_-}{2}.
\end{equation} 
\end{theorem}
Note that we are omitting the underlying measure $\sigma$ in the definition of $V_{\lambda_e,\lambda_n}$ in \eqref{pre eq2}, as we already did in \eqref{defi V} for the generic potential $V$. We will keep this abuse of notation in the definition of all shell potentials appearing throughout the article.

\begin{proof}
We first prove that $V=V_{\lambda_e,\lambda_n}$ on $D(T_{\Lambda_{\lambda_e,\lambda_n}})$, which would imply that $T_{\Lambda_{\lambda_e,\lambda_n}}=H+V_{\lambda_e,\lambda_n}$ on $D(T_{\Lambda_{\lambda_e,\lambda_n}})$.
Let $\varphi=\Phi(G+g)\in D(T_{\Lambda_{\lambda_e,\lambda_n}})$. Then $\varphi_\pm=\Phi_\sigma G +(\mp\frac{i}{2}(\alpha\cdot N)+C_\sigma)g$ by Lemma \ref{l jump}$(i)$, which yields $$\frac{\varphi_++\varphi_-}{2}=\Phi_\sigma G +C_\sigma g.$$ Recall that $\Phi_\sigma G=\Lambda_{\lambda_e,\lambda_n}g$ by \eqref{def T_Lambda}, hence
\begin{equation*}
\begin{split}
V_{\lambda_e,\lambda_n}\varphi
&=\{\lambda_e+\lambda_n(\alpha\cdot N)\}(\Phi_\sigma G +C_\sigma g)\\
&=\{\lambda_e+\lambda_n(\alpha\cdot N)\}\frac{\lambda_n(\alpha\cdot N)-\lambda_e}{\lambda_e^2-\lambda_n^2}\,g=-g=V\varphi,
\end{split}
\end{equation*}
by \eqref{defi V}, and we are done.

The first statement in the theorem follows by Theorem \ref{pre t1}, as far as we check that $\Lambda_{\lambda_e,\lambda_n}$ is self-adjoint and Fredholm. 
Since $\lambda_e,\lambda_n\in\R$, the self-adjointness of $\Lambda_{\lambda_e,\lambda_n}$ follows from the one of $C_\sigma$ and $\alpha\cdot N$. It remains to show that  $\Lambda_{\lambda_e,\lambda_n}$ is Fredholm if $\lambda_e^2-\lambda_n^2\neq4$, which follows by arguments similar to the ones in the proof of \cite[Theorem 5.5]{AMV2}. Set 
\begin{equation*}
\Lambda_\pm=\frac{\mp\lambda_n(\alpha\cdot N)+\lambda_e}{\lambda_e^2-\lambda_n^2}\pm C_\sigma.
\end{equation*} 
Note that $C_\sigma^2=C_\sigma(\alpha\cdot N)(\alpha\cdot N)C_\sigma=C_\sigma(\alpha\cdot N)\{C_\sigma,\alpha\cdot N\}+1/4$ by Lemma \ref{l jump}$(ii)$, where $\{C_\sigma,\alpha\cdot N\}$ denotes the anticommutator $C_\sigma(\alpha\cdot N)+(\alpha\cdot N)C_\sigma$, so we get
\begin{equation*}
\begin{split}
\Lambda_+\Lambda_-=\frac{1}{\lambda_e^2-\lambda_n^2}-\frac{1}{4}+\bigg(\frac{\lambda_n}{\lambda_e^2-\lambda_n^2}-C_\sigma(\alpha\cdot N)\bigg)\{C_\sigma,\alpha\cdot N\}.
\end{split}
\end{equation*} 
Since $\partial\Omega$ is smooth, $\{C_\sigma,\alpha\cdot N\}$ is a compact operator in $L^2(\sigma)^4$ by \cite[Lemma 3.5]{AMV1}. Hence, $\Lambda_-\Lambda_+=\Lambda_+\Lambda_-$ is Fredholm for $\lambda_e^2-\lambda_n^2\neq4$, and thus $\Lambda_\pm$ are also Fredholm operators by \cite[Theorem 1.46$(iii)$]{Aiena}. Since $\Lambda_{\lambda_e,\lambda_n}=-\Lambda_+$, we get that $\Lambda_{\lambda_e,\lambda_n}$ is Fredholm, and the proof of the theorem is complete thanks to Theorem \ref{pre t1}.
\end{proof}

\begin{lemma}\label{propo elec}
Let $\lambda_e,\lambda_n\in\R$ be such that $\lambda_e^2-\lambda_n^2\neq0$. If $\theta\in\R$ is such that 
\begin{equation}\label{invert l}
\frac{1}{2}(\lambda_e^2-\lambda_n^2+4)(1+\cos\theta)-4
+2\lambda_n\sin\theta\neq0,
\end{equation}
then the pair $(\Lambda_{\lambda_e,\lambda_n},(\Lambda_{\lambda_e,\lambda_n})_z)$ satisfies \eqref{uni equiv eq1}, where $z=e^{i\theta}$, $\Lambda_{\lambda_e,\lambda_n}$ is defined by \eqref{pre eq1} and
\begin{equation}\label{pre eq3}
\begin{split}
(\Lambda_{\lambda_e,\lambda_n})_z&=\frac{\big(\lambda_n\cos\theta
-\frac{1}{4}(\lambda_e^2-\lambda_n^2+4)
\sin\theta\big)(\alpha\cdot N)-\lambda_e}
{\frac{1}{2}(\lambda_e^2-\lambda_n^2+4)(1+\cos\theta)-4
+2\lambda_n\sin\theta}-C_\sigma.
\end{split}
\end{equation}
\end{lemma}

\begin{proof}
For shortness sake, set $$\lambda_1=\frac{-\lambda_e}{\lambda_e^2-\lambda_n^2}
\quad\text{and}\quad
\lambda_2=\frac{\lambda_n}{\lambda_e^2-\lambda_n^2},
\quad\text{so}\quad
\lambda_1^2-\lambda_2^2=\frac{1}{\lambda_e^2-\lambda_n^2}.$$
Then $\Lambda_{\lambda_e,\lambda_n}=\lambda_1+\lambda_2(\alpha\cdot N)-C_\sigma$. 
We are going to find the formula for $(\Lambda_{\lambda_e,\lambda_n})_z$ by working with the pair $(\Lambda_{\lambda_e,\lambda_n},(\Lambda_{\lambda_e,\lambda_n})_z)$ in \eqref{uni equiv eq1}. On one hand, 
\begin{equation}\label{elec eq1}
\begin{split}
\frac{1+z}{2}+(1-z)i(\alpha\cdot N)(\Lambda_{\lambda_e,\lambda_n}+C_\sigma)
&=\frac{1+z}{2}+(1-z)i(\alpha\cdot N)(\lambda_1+\lambda_2(\alpha\cdot N))\\
&=\frac{1+z}{2}+\lambda_2(1-z)i+\lambda_1(1-z)i(\alpha\cdot N).
\end{split}
\end{equation}
On the other hand, using Lemma \ref{l jump}$(ii)$,
\begin{equation}\label{elec eq2}
\begin{split}
\bigg(\frac{1+z}{2}-(1-&z)iC_\sigma(\alpha\cdot N)\bigg)\Lambda_{\lambda_e,\lambda_n}\\
&=\bigg(\frac{1+z}{2}-(1-z)iC_\sigma(\alpha\cdot N)\bigg)(\lambda_1+\lambda_2(\alpha\cdot N)-C_\sigma)\\
&=\lambda_1\frac{1+z}{2}+\bigg(\lambda_2\frac{1+z}{2}-\frac{1-z}{4}i\bigg)(\alpha\cdot N)\\
&\quad-\lambda_1(1-z)iC_\sigma(\alpha\cdot N)-
\bigg(\frac{1+z}{2}+\lambda_2(1-z)i\bigg)C_\sigma.
\end{split}
\end{equation} 
Note also that 
\begin{equation}\label{elec eq3}
\begin{split}
\bigg(\frac{1+z}{2}+\lambda_2(1-z)i+\lambda_1(1-z&)i(\alpha\cdot N)\bigg)
\bigg(\frac{1+z}{2}+\lambda_2(1-z)i-\lambda_1(1-z)i(\alpha\cdot N)\bigg)\\
&=\frac{(1+z)^2}{4}+(\lambda_1^2-\lambda_2^2)(1-z)^2+\lambda_2(1-z^2)i.
\end{split}
\end{equation}
If we apply the operator $\frac{1+z}{2}+\lambda_2(1-z)i-\lambda_1(1-z)i(\alpha\cdot N)$ from the right in \eqref{uni equiv eq1} to the pair $(\Lambda_{\lambda_e,\lambda_n},(\Lambda_{\lambda_e,\lambda_n})_z)$, using \eqref{elec eq1}, \eqref{elec eq3} and \eqref{elec eq2}, we deduce that
\begin{equation}\label{aa eq1}
\begin{split}
\bigg(\frac{(1+z)^2}{4}&+(\lambda_1^2-\lambda_2^2)(1-z)^2+\lambda_2(1-z^2)i\bigg)(\Lambda_{\lambda_e,\lambda_n})_z\\
&=\bigg\{\lambda_1\frac{1+z}{2}+\bigg(\lambda_2\frac{1+z}{2}-\frac{1-z}{4}i\bigg)(\alpha\cdot N)
-\lambda_1(1-z)iC_\sigma(\alpha\cdot N)\\
&\quad-\bigg(\frac{1+z}{2}+\lambda_2(1-z)i\bigg)C_\sigma\bigg\}
\bigg\{\frac{1+z}{2}+\lambda_2(1-z)i-\lambda_1(1-z)i(\alpha\cdot N)\bigg\}\\
&=\lambda_1z+\bigg\{\bigg(\lambda_2^2-\lambda_1^2-\frac{1}{4}\bigg)\frac{1-z^2}{2}i+\lambda_2\frac{1+z^2}{2}\bigg\}(\alpha\cdot N)\\
&\quad-\bigg(\frac{(1+z)^2}{4}+(\lambda_1^2-\lambda_2^2)(1-z)^2+\lambda_2(1-z^2)i\bigg)C_\sigma.
\end{split}
\end{equation} 
Since $|z|=1$, we can divide \eqref{aa eq1} by $z$ to get
\begin{equation}\label{aa eq2}
\begin{split}
\bigg(\frac{(1+z)^2}{4z}&+(\lambda_1^2-\lambda_2^2)\frac{(1-z)^2}{z}+\lambda_2\frac{1-z^2}{z}i\bigg)(\Lambda_{\lambda_e,\lambda_n})_z\\
&=\lambda_1+\bigg\{\bigg(\lambda_2^2-\lambda_1^2-\frac{1}{4}\bigg)\frac{1-z^2}{2z}i+\lambda_2\frac{1+z^2}{2z}\bigg\}(\alpha\cdot N)\\
&\quad-\bigg(\frac{(1+z)^2}{4z}+(\lambda_1^2-\lambda_2^2)\frac{(1-z)^2}{z}+\lambda_2\frac{1-z^2}{z}i\bigg)C_\sigma.
\end{split}
\end{equation} 
For $z=e^{i\theta}$ with $\theta\in\R$, we easily see that
$$\frac{(1\pm z)^2}{z}=2(\cos\theta\pm1),
\quad\frac{1-z^2}{z}i=2\sin\theta,
\quad\frac{1+z^2}{z}=2\cos\theta.$$
Therefore, \eqref{aa eq2} is equivalent to
\begin{equation*}
\begin{split}
\bigg\{1+2(\cos\theta-&1)\bigg(\lambda_1^2-\lambda_2^2+\frac{1}{4}\bigg)+2\lambda_2\sin\theta\bigg\}(\Lambda_{\lambda_e,\lambda_n})_z\\
&=\lambda_1+\bigg\{\lambda_2\cos\theta-\bigg(\lambda_1^2-\lambda_2^2+\frac{1}{4}\bigg)\sin\theta\bigg\}(\alpha\cdot N)\\
&\quad-\bigg\{1+2(\cos\theta-1)\bigg(\lambda_1^2-\lambda_2^2+\frac{1}{4}\bigg)+2\lambda_2\sin\theta\bigg\}C_\sigma
\end{split}
\end{equation*}
which, in terms of $\lambda_e$ and $\lambda_n$ (using that $\lambda_e^2-\lambda_n^2\neq0$), corresponds to
\begin{equation}\label{aa eq3}
\begin{split}
\bigg\{\frac{1}{2}(\lambda_e^2-\lambda_n^2&+4)(1+\cos\theta)-4
+2\lambda_n\sin\theta\bigg\}(\Lambda_{\lambda_e,\lambda_n})_z\\
&=-\lambda_e+\bigg\{\lambda_n\cos\theta
-\frac{1}{4}(\lambda_e^2-\lambda_n^2+4)
\sin\theta\bigg\}(\alpha\cdot N)\\
&\quad-\bigg\{\frac{1}{2}
(\lambda_e^2-\lambda_n^2+4)(1+\cos\theta)-4
+2\lambda_n\sin\theta\bigg\}C_\sigma,
\end{split}
\end{equation}
The lemma follows by \eqref{aa eq3} and \eqref{invert l} because all the computations above can be reverted if $\lambda_e^2-\lambda_n^2\neq0$ and \eqref{invert l} holds.
\end{proof}

Note that, from \eqref{pre eq3}, $(\Lambda_{\lambda_e,\lambda_n})_1=\Lambda_{\lambda_e,\lambda_n}$ for all $\lambda_e^2-\lambda_n^2\neq0$. Of course, this comes as no surprise because 
$(T_{\Lambda_{\lambda_e,\lambda_n}})_1=T_{\Lambda_{\lambda_e,\lambda_n}}$ by \eqref{uni equiv def}, since $U_1=1$. 

\begin{theorem}\label{main apli theo}
Given $\lambda_e,\lambda_n,\theta\in\R$ set 
\begin{equation}\label{pre eq4}
\begin{split}
\gamma&=\frac{1}{2}(\lambda_e^2-\lambda_n^2+4)(1+\cos\theta)
-4+2\lambda_n\sin\theta,\\
\lambda_n'&=\lambda_n\cos\theta
-\frac{1}{4}(\lambda_e^2-\lambda_n^2+4)
\sin\theta.
\end{split}
\end{equation}
Assume that $\lambda_e^2-\lambda_n^2\neq 0,4$, that $\gamma\neq0$ and that $\lambda_e^2-\lambda_n'^2\neq 0,\gamma^2/4$. Then $$T_{\Lambda_{\lambda_e,\lambda_n}}\text{ and } T_{\Lambda_{\gamma\lambda_e/(\lambda_e^2-\lambda_n'^2),
\gamma\lambda_n'/(\lambda_e^2-\lambda_n'^2)}},\text{ both defined by \eqref{def T_Lambda} and \eqref{pre eq1}},$$ are unitary equivalent self-adjoint operators. Moreover, 
\begin{equation*}
\begin{split}
T_{\Lambda_{\lambda_e,\lambda_n}}&=H+V_{\lambda_e,\lambda_n}\text{ and}\\
T_{\Lambda_{\gamma\lambda_e/(\lambda_e^2-\lambda_n'^2),
\gamma\lambda_n'/(\lambda_e^2-\lambda_n'^2)}}
&=H+V_{\gamma\lambda_e/(\lambda_e^2-\lambda_n'^2),
\gamma\lambda_n'/(\lambda_e^2-\lambda_n'^2)}
\end{split}
\end{equation*} 
on $D\big(T_{\Lambda_{\lambda_e,\lambda_n}}\big)$ and $D\big(T_{\Lambda_{\gamma\lambda_e/(\lambda_e^2-\lambda_n'^2),
\gamma\lambda_n'/(\lambda_e^2-\lambda_n'^2)}}\big)$, respectively, where $$V_{\lambda_e,\lambda_n}\text{ and } V_{\gamma\lambda_e/(\lambda_e^2-\lambda_n'^2),
\gamma\lambda_n'/(\lambda_e^2-\lambda_n'^2)}\text{ are defined by \eqref{pre eq2}}.$$
\end{theorem}
\begin{proof}
The theorem follows by a combination of Theorem \ref{t elec1} and Lemmata \ref{propo elec} and \ref{lema elec}. 
On one hand, since $\lambda_e^2-\lambda_n^2\neq 0,4$, from Theorem \ref{t elec1} we see that $T_{\Lambda_{\lambda_e,\lambda_n}}$ is self-adjoint and $T_{\Lambda_{\lambda_e,\lambda_n}}=H+V_{\lambda_e,\lambda_n}$ on
$D(T_{\Lambda_{\lambda_e,\lambda_n}})$. 

On the other hand, that $\gamma\neq0$ means that \eqref{invert l} holds. Hence, Lemma \ref{propo elec} shows that the pair $(\Lambda_{\lambda_e,\lambda_n},(\Lambda_{\lambda_e,\lambda_n})_z)$ satisfies \eqref{uni equiv eq1}, where $z=e^{i\theta}$ and $(\Lambda_{\lambda_e,\lambda_n})_z$ is defined by \eqref{pre eq3}. 
Since $\lambda_e^2-\lambda_n'^2\neq 0$, we can set
$$\lambda_1=\frac{\gamma\lambda_e}{\lambda_e^2-\lambda_n'^2}\quad\text{and}\quad
\lambda_2=\frac{\gamma\lambda_n'}{\lambda_e^2-\lambda_n'^2},\quad\text{so}\quad \lambda_1^2-\lambda_2^2=\frac{\gamma^2}{\lambda_e^2-\lambda_n'^2}.$$ Then,
from \eqref{pre eq3}, \eqref{pre eq4} and \eqref{pre eq1}, we easily get
\begin{equation*}
\begin{split}
(\Lambda_{\lambda_e,\lambda_n})_z
&=\frac{\lambda_n'(\alpha\cdot N)-\lambda_e}{\gamma}-C_\sigma
=\frac{\lambda_2(\alpha\cdot N)-\lambda_1}{\lambda_1^2-\lambda_2^2}-C_\sigma\\
&=\Lambda_{\lambda_1,\lambda_2}
=\Lambda_{\gamma\lambda_e/(\lambda_e^2-\lambda_n'^2),\gamma\lambda_n'/(\lambda_e^2-\lambda_n'^2)}.
\end{split}
\end{equation*}
Now, since the pair $(\Lambda_{\lambda_e,\lambda_n},(\Lambda_{\lambda_e,\lambda_n})_z)$ satisfies \eqref{uni equiv eq1}, Lemma \ref{lema elec} yields that 
$$(T_{\Lambda_{\lambda_e,\lambda_n}})_z\subset T_{(\Lambda_{\lambda_e,\lambda_n})_z}
=T_{\Lambda_{\gamma\lambda_e/(\lambda_e^2-\lambda_n'^2),\gamma\lambda_n'/(\lambda_e^2-\lambda_n'^2)}}.$$
Besides, Theorem \ref{t elec1} shows that $T_{\Lambda_{\gamma\lambda_e/(\lambda_e^2-\lambda_n'^2),\gamma\lambda_n'/(\lambda_e^2-\lambda_n'^2)}}$ is self-adjoint whenever 
$$\bigg(\frac{\gamma\lambda_e}{\lambda_e^2-\lambda_n'^2}\bigg)^2-
\bigg(\frac{\gamma\lambda_n'}{\lambda_e^2-\lambda_n'^2}\bigg)^2\neq0,4,$$
and this last relation holds if $\gamma\neq0$ and $\lambda_e^2-\lambda_n'^2\neq0,\gamma^2/4$. In conclusion, from the assumptions in the statement of the theorem, we have proven that $$(T_{\Lambda_{\lambda_e,\lambda_n}})_z\subset 
T_{\Lambda_{\gamma\lambda_e/(\lambda_e^2-\lambda_n'^2),\gamma\lambda_n'/(\lambda_e^2-\lambda_n'^2)}}$$ and that both operators are self-adjoint (the first one due to the fact that, by construction, it is unitary equivalent to $T_{\Lambda_{\lambda_e,\lambda_n}}$, which is self-adjoint). Therefore, both operators coincide, and thus $T_{\Lambda_{\lambda_e,\lambda_n}}$ and $T_{\Lambda_{\gamma\lambda_e/(\lambda_e^2-\lambda_n'^2),\gamma\lambda_n'/(\lambda_e^2-\lambda_n'^2)}}$ are unitary equivalent, as claimed. Finally, that
\begin{equation*}
\begin{split}
T_{\Lambda_{\gamma\lambda_e/(\lambda_e^2-\lambda_n'^2),
\gamma\lambda_n'/(\lambda_e^2-\lambda_n'^2)}}
=H+V_{\gamma\lambda_e/(\lambda_e^2-\lambda_n'^2),
\gamma\lambda_n'/(\lambda_e^2-\lambda_n'^2)}
\end{split}
\end{equation*}
on $D\big(T_{\Lambda_{\gamma\lambda_e/(\lambda_e^2-\lambda_n'^2),
\gamma\lambda_n'/(\lambda_e^2-\lambda_n'^2)}}\big)$ is also given by Theorem \ref{t elec1}. 
\end{proof}

The following corollaries are the applications of Theorem \ref{main apli theo} mentioned in the introduction. Theorems \ref{coro1 inf}, \ref{coro2 inf} and \ref{coro3 inf} correspond to informal statemets of Corollaries \ref{coro1}, \ref{coro2} and \ref{coro3}, respectively. 

\begin{corollary}\label{coro1}
Let $\lambda_e,\lambda_n\in\R$ be such that $\lambda_e^2-\lambda_n^2\neq 0,4$. 
Then $$T_{\Lambda_{\lambda_e,\lambda_n}}\text{ and } T_{\Lambda_{-4\lambda_e/(\lambda_e^2-\lambda_n^2),
4\lambda_n/(\lambda_e^2-\lambda_n^2)}},\text{ both defined by \eqref{def T_Lambda} and \eqref{pre eq1}},$$ are unitary equivalent self-adjoint operators. Moreover, $T_{\Lambda_{\lambda_e,\lambda_n}}=H+V_{\lambda_e,\lambda_n}$ and
\begin{equation*}
\begin{split}
T_{\Lambda_{-4\lambda_e/(\lambda_e^2-\lambda_n^2),
4\lambda_n/(\lambda_e^2-\lambda_n^2)}}
=H+V_{-4\lambda_e/(\lambda_e^2-\lambda_n^2),
4\lambda_n/(\lambda_e^2-\lambda_n^2)}
\end{split}
\end{equation*}
on $D(T_{\Lambda_{\cdot,\cdot}})$, where $V_{\cdot,\cdot}$ is given by \eqref{pre eq2}.
\end{corollary}

\begin{proof}
Apply Theorem \ref{main apli theo} taking $\theta=\pi$.
\end{proof}

In particular, from Corollary \ref{coro1} we get that $H+V_{\lambda,0}$ and $H+V_{-4/\lambda,0}$ (defined on $D(T_{\Lambda_{\lambda,0}})$ and $D(T_{\Lambda_{-4/\lambda,0}})$, respectively) are unitary equivalent self-adjoint operators for all $\lambda^2\neq0,4$, which strengthens the first conclusion in \cite[Theorem 3.3]{AMV2}. As a byproduct, we also get that $H+V_{0,\lambda}$ and $H+V_{0,-4/\lambda}$ are unitary equivalent self-adjoint operators for all $\lambda\neq0$.

\begin{corollary}\label{coro2}
Let $\lambda_e,\lambda_n\in\R$ be such that $\lambda_e\neq0$ and $\lambda_e^2-\lambda_n^2=-4$. 
Then $$T_{\Lambda_{\lambda_e,\lambda_n}}\text{ and } T_{\Lambda_{(\pm2\lambda_n-4)/\lambda_e,0}},\text{ all of them defined by \eqref{def T_Lambda} and \eqref{pre eq1}},$$ are unitary equivalent self-adjoint operators. Moreover, $T_{\Lambda_{\lambda_e,\lambda_n}}=H+V_{\lambda_e,\lambda_n}$ and
\begin{equation*}
\begin{split}
T_{\Lambda_{(\pm2\lambda_n-4)/\lambda_e,0}}
=H+V_{(\pm2\lambda_n-4)/\lambda_e,0}
\end{split}
\end{equation*}
on $D(T_{\Lambda_{\cdot,\cdot}})$, where $V_{\cdot,\cdot}$ is given by \eqref{pre eq2}.
\end{corollary}
\begin{proof}
Apply Theorem \ref{main apli theo} taking $\theta=\pm\pi/2$.
\end{proof}

\begin{corollary}\label{coro3}
Let $\lambda_e,\lambda_n\in\R\setminus\{0\}$ be such that
$|\lambda_e^2-\lambda_n^2|\neq0,4$. Assume that
\begin{equation}\label{cond exist theta}
\lambda_e^2-\lambda_n^2+2\lambda_e\neq4
\quad\text{and}\quad\lambda_e^2-\lambda_n^2-2\lambda_e\neq4.
\end{equation}
Then, there exists $\theta\in\R$ so that
\begin{equation}\label{cond theta}
\begin{split}
\tan\theta&=\frac{4\lambda_n}{\lambda_e^2-\lambda_n^2+4},\quad\text{and}\\
\cos\theta&\neq\frac{16-(\lambda_e^2-\lambda_n^2)^2}
{16+(\lambda_e^2-\lambda_n^2)^2+8(\lambda_e^2+\lambda_n^2)},
\frac{16-(\lambda_e^2-\lambda_n^2)^2
\pm4\lambda_e(\lambda_e^2-\lambda_n^2+4)}
{16+(\lambda_e^2-\lambda_n^2)^2+8(\lambda_e^2+\lambda_n^2)}.
\end{split}
\end{equation}

For any $\theta\in\R$ as in \eqref{cond theta},
$$T_{\Lambda_{\lambda_e,\lambda_n}}\text{ and } T_{\Lambda_{\left(2\lambda_n\frac{1+\cos\theta}{\sin\theta}-4\right)/\lambda_e,0}},\text{ both defined by \eqref{def T_Lambda} and \eqref{pre eq1}},$$ are unitary equivalent self-adjoint operators. Moreover, $T_{\Lambda_{\lambda_e,\lambda_n}}=H+V_{\lambda_e,\lambda_n}$ and
\begin{equation*}
\begin{split}
T_{\Lambda_{\left(2\lambda_n\frac{1+\cos\theta}{\sin\theta}-4\right)/\lambda_e,0}}
=H+V_{\left(2\lambda_n\frac{1+\cos\theta}{\sin\theta}-4\right)/\lambda_e,0}
\end{split}
\end{equation*}
on $D(T_{\Lambda_{\cdot,\cdot}})$, where $V_{\cdot,\cdot}$ is given by \eqref{pre eq2}.
\end{corollary}
\begin{proof}
This is also a consequence of Theorem \ref{main apli theo}. More precisely, a (tedious) computation shows that if $\lambda_e,\lambda_n\in\R\setminus\{0\}$ satisfy 
$|\lambda_e^2-\lambda_n^2|\neq0,4$ and if there exists $\theta\in\R$ such that \eqref{cond theta} holds, then the assuptions in the statement of Theorem \ref{main apli theo} are fulfilled, and $\lambda_n'=0$ in \eqref{pre eq4}. We leave the details of this part for the reader.

We impose \eqref{cond exist theta} to show the existence of a $\theta\in\R$ which satisfies \eqref{cond theta}. To see this, set
$$p=16-(\lambda_e^2-\lambda_n^2)^2\quad\text{and}\quad
q=4\lambda_e(\lambda_e^2-\lambda_n^2+4).$$ 
Note that $p,q\neq0$ by assumption, thus $|p+q|\neq|p-q|$.
If $|p\pm q|=|p|$ then we must have 
\begin{equation*}
\begin{split}
0&=\pm q+2p=\pm4\lambda_e(\lambda_e^2-\lambda_n^2+4)
+32-2(\lambda_e^2-\lambda_n^2)^2\\
&=\pm4\lambda_e(\lambda_e^2-\lambda_n^2+4)
-2(\lambda_e^2-\lambda_n^2+4)(\lambda_e^2-\lambda_n^2-4),
\end{split}
\end{equation*}
which is equivalent to $0=\pm2\lambda_e-(\lambda_e^2-\lambda_n^2-4)$ because $\lambda_e^2-\lambda_n^2\neq-4$ by assumption. Hence, $|p\pm q|\neq|p|$ if \eqref{cond exist theta} holds. That is, under the assumptions of the corollary, we have seen that 
\begin{equation}\label{aux 1}
p-q,\,p\text{ and }p+q\text{ have different absolute value each other.}
\end{equation}
Now, let $\theta\in\R$ be such that 
$\tan\theta=4\lambda_n/(\lambda_e^2-\lambda_n^2+4).$ If $\theta$ fulfills \eqref{cond theta} then we are done. If not, it means that $\cos\theta$ coincides with one of the three terms on the right hand side of the second condition in  \eqref{cond theta}. But then, by \eqref{aux 1}, it is enough to pick $\theta+\pi$ instead of $\theta$, since $\tan(\theta+\pi)=\tan\theta$ but $\cos(\theta+\pi)=-\cos\theta$, and so $\theta+\pi$ fulfills \eqref{cond theta}. 
\end{proof}

\section{Magnetic shell potentials}\label{s nonc magnetic}
This section concerns the construction of shell interactions for the free Dirac operator with regular potentials on $\partial\Omega$ of magnetic type. 

\begin{proposition}\label{propo magnetic}
Let $\lambda:\partial\Omega\to\R$ be of class $\CC^1(\partial\Omega)$ and such that $\lambda(x)\neq0$ for all $x\in\partial\Omega$. Set
\begin{equation}\label{l8eq10}
\Lambda_{\lambda}=-\frac{1}{\lambda}(\alpha\cdot N)-C_\sigma,
\end{equation}
where $C_\sigma$ is as in {\em Lemma \ref{l jump}}.
Then $T_{\Lambda_{\lambda}}$ given by \eqref{def T_Lambda} is self-adjoint and $T_{\Lambda_{\lambda}}=H+\VV_{\lambda}$ on $D(T_{\Lambda_{\lambda}})$, where $\VV_{\lambda}$ is defined, for $\varphi\in \Phi(\XX)$ and $\varphi_\pm$ as in {\em Lemma \ref{l jump}}, by
\begin{equation}\label{l8eq11}
\VV_{\lambda}\varphi=\lambda(\alpha\cdot N)\frac{\varphi_++\varphi_-}{2}.
\end{equation} 
\end{proposition}

\begin{proof}
The proof follows exactly the same lines as the one of Theorem \ref{t elec1}. To see that $V=\VV_{\lambda}$ on $D(T_{\Lambda_{\lambda}})$, let
$\varphi=\Phi(G+g)\in D(T_{\Lambda_{\lambda}})$. Then $\varphi_\pm=\Phi_\sigma G +(\mp\frac{i}{2}(\alpha\cdot N)+C_\sigma)g$ by Lemma \ref{l jump}$(i)$, which yields
$\varphi_++\varphi_-=2(\Phi_\sigma G +C_\sigma g).$ Recall that $\Phi_\sigma G=\Lambda_{\lambda} g$ by \eqref{def T_Lambda}, hence
\begin{equation*}
\begin{split}
\VV_{\lambda}\varphi
&=\lambda(\alpha\cdot N)(\Phi_\sigma G +C_\sigma g)
=-\lambda(\alpha\cdot N)\frac{1}{\lambda}(\alpha\cdot N)g
=-g=V\varphi,
\end{split}
\end{equation*}
where we used \eqref{defi V} in the last equality above.

Since $\lambda$ is continuous on the compact set $\partial\Omega$ and $\lambda(x)\neq0$ for all $x\in\partial\Omega$, there exists $\epsilon>0$ such that 
\begin{equation}\label{hyp magnetic}
\epsilon\leq\lambda(x)^2\leq 1/\epsilon\quad\text{for all } x\in\partial\Omega.
\end{equation}
Hence $\Lambda_{\lambda}$ is a bounded operator in $L^2(\sigma)^2$. The first statement in the proposition follows by Theorem \ref{pre t1}, as far as we check that $\Lambda_{\lambda}$ is self-adjoint and Fredholm. 
Since $\lambda$ is real-valued, the self-adjointness of $\Lambda_{\lambda}$ follows from the one of $C_\sigma$ and $\alpha\cdot N$. Regarding the Fredholm character, recall that $$C_\sigma^2=C_\sigma(\alpha\cdot N)\{C_\sigma,\alpha\cdot N\}+1/4$$ by Lemma \ref{l jump}$(ii)$, where $\{C_\sigma,\alpha\cdot N\}=C_\sigma(\alpha\cdot N)+(\alpha\cdot N)C_\sigma$. 
Then
\begin{equation}\label{l8eq2}
\begin{split}
\Lambda_{\lambda}^2&=\frac{1}{\lambda^2}+C_\sigma^2
+\bigg\{\frac{1}{\lambda}(\alpha\cdot N),C_\sigma\bigg\}\\
&=\frac{1}{\lambda^2}+\frac{1}{4}
+C_\sigma(\alpha\cdot N)\{C_\sigma,\alpha\cdot N\}
+\bigg\{\frac{1}{\lambda}(\alpha\cdot N),C_\sigma\bigg\}.
\end{split}
\end{equation} 
By \eqref{hyp magnetic}, $1/\lambda^2+1/4$ is bounded and invertible in $L^2(\sigma)^4$. Besides, $\{C_\sigma,\alpha\cdot N\}$ is  compact in $L^2(\sigma)^4$ by \cite[Lemma 3.5]{AMV1} because $\partial\Omega$ is smooth. The compactness of $\{\frac{1}{\lambda}(\alpha\cdot N),C_\sigma\}$ follows similarly. More precisely, arguing as in the proof of \cite[Lemma 3.5]{AMV1} one shows that 
$\{\frac{1}{\lambda}(\alpha\cdot N),C_\sigma\}g(x)
=\lim_{\epsilon\searrow0}\int_{|x-z|>\epsilon}
K(x,z)g(z)\,d\sigma(z)$ for all $g\in L^2(\sigma)^4$, where
\begin{equation*}
\begin{split}
K(x,z)&=\phi(x-z)\,\alpha\cdot \bigg(\frac{N(z)}{\lambda(z)}-\frac{N(x)}{\lambda(x)}\bigg)
+\frac{ie^{-m|x-z|}}{2\pi|x-z|^{3}}\,(1+m|x-z|)
\bigg(\frac{N(x)}{\lambda(x)}\cdot(x-z)\bigg).
\end{split}
\end{equation*}
The smoothness of $\partial\Omega$ and $\lambda$ guarantee that $K(x,z)=O(|x-z|^{-1})$ for $|x-z|\to0$, which yields the compactness of $\{\frac{1}{\lambda}(\alpha\cdot N),C_\sigma\}$ because $\partial\Omega$ is bounded.
Therefore, from \eqref{l8eq2} we see that $\Lambda^2_\lambda$ is Fredholm, and thus $\Lambda_\lambda$ is also Fredholm by \cite[Theorem 1.46$(iii)$]{Aiena}. The proof of the proposition is complete thanks to Theorem \ref{pre t1}.
\end{proof}

\begin{remark}
The reader may realize that less regularity on $\partial\Omega$ and $\lambda$ can be assumed to get the same conclusions of Proposition \ref{propo magnetic}. For instance, $N/\lambda\in\CC^\alpha(\partial\Omega)$ with $\alpha>0$ would suffice, and even less. However, we don't want to look for optimal regularity assumptions to avoid technicalities.
\end{remark}

The purpose of the following theorem is to get rid of the non-vanishing assumption of $\lambda$ in Proposition \ref{propo magnetic}. In order to do so, we use a generalization of the unitary transformations introduced at the beginning of Section \ref{s uni equiv}: 
\begin{equation*}
\text{Given $\theta:\R^3\to\R$, define
$U_\theta:L^2(\mu)^4\to L^2(\mu)^4$ by $U_\theta\varphi=e^{-i\theta\chi_{\Omega_-}}\varphi.$}
\end{equation*}

\begin{theorem}\label{t nonc magnetic}
Let $\lambda:\R^3\to\R$ be a compactly supported function of class $\CC^1(\R^3)$ and let $\lambda_{\partial\Omega}$ denote its restriction to $\partial\Omega$. Let $M>0$ be such that $\lambda_{\partial\Omega}(x)+M>0$ for all $x\in\partial\Omega$ and let $\theta:\R^3\to\R$ of class $\CC^1(\R^3)$ so that
\begin{equation}\label{xyz1}
e^{i\theta}=\frac{(\lambda+2i)(\lambda+M-2i)}{(\lambda-2i)(\lambda+M+2i)}.
\end{equation}
Then, the operator 
\begin{equation}\label{xyz3}
\big(T_{\Lambda_{\lambda_{\partial\Omega}+M}}\big)_\theta
=U_\theta^{*}T_{\Lambda_{\lambda_{\partial\Omega}+M}}U_\theta
+\chi_{\Omega_-}(\alpha\cdot\nabla\theta)
\end{equation}
defined on $U_\theta^{*}D\big(T_{\Lambda_{\lambda_{\partial\Omega}+M}}\big)$ is self-adjoint, where $\Lambda_{\lambda_{\partial\Omega}+M}$ and $T_{\Lambda_{\lambda_{\partial\Omega}+M}}$ are given by \eqref{l8eq10} and \eqref{def T_Lambda}, respectively. Moreover, 
\begin{equation}\label{xyz2}
\big(T_{\Lambda_{\lambda_{\partial\Omega}+M}}\big)_\theta
=H+\VV_{\lambda_{\partial\Omega}}\text{ on }\, U_\theta^{*}D\big(T_{\Lambda_{\lambda_{\partial\Omega}+M}}\big),
\end{equation}
where $\VV_{\lambda_{\partial\Omega}}$ is given by \eqref{l8eq11}.
\end{theorem}

\begin{proof}
The existence of $M$ in the statement of the theorem follows from the fact that $\lambda_{\partial\Omega}$ is continuous on the compact set $\partial\Omega$, and the existence of $\theta$ is a consequence of the fact that the term on the right hand side of \eqref{xyz1} is a complex-valued function of modulus 1. Furthermore, we can take $\theta$ so that it inherits the regularity of $\lambda$. To see this, note that 
$$\frac{(\lambda(x)+2i)(\lambda(x)+M-2i)}{(\lambda(x)-2i)(\lambda(x)+M+2i)}\neq1$$
for all $x\in \R^3$ because $t\mapsto(t+2i)/(t-2i)$ is injective from $\R$ to the unit circle in $\R^2$, $M>0$ and $\lambda$ can not take the values $\pm\infty$. Hence, there exists a well defined principal value of the argument of the points given by the right hand side of \eqref{xyz1} which is as smooth as $\lambda$ (thanks to the holomorphicity of the corresponding branch of the complex logarithm).

Since ${\lambda_{\partial\Omega}+M}>0$ in $\partial\Omega$, from Proposition \ref{propo magnetic} we deduce that $T_{\Lambda_{\lambda_{\partial\Omega}+M}}$ is self-adjoint and $T_{\Lambda_{\lambda_{\partial\Omega}+M}}
=H+\VV_{\lambda_{\partial\Omega}+M}$ on $D\big(T_{\Lambda_{\lambda_{\partial\Omega}+M}}\big)$, 
where $\VV_{\lambda_{\partial\Omega}+M}$ is given by \eqref{l8eq11}. In particular, $\VV_{\lambda_{\partial\Omega}+M}=V$ on $D\big(T_{\Lambda_{\lambda_{\partial\Omega}+M}}\big)$, for $V$ as in \eqref{defi V}. Furthermore, thanks to Lemma \ref{l jump}$(i)$ this means that if $\varphi=\Phi(G+g)\in D\big(T_{\Lambda_{\lambda_{\partial\Omega}+M}}\big)$ then
\begin{equation*}
i(\alpha\cdot N)(\varphi_--\varphi_+)=-g=V\varphi
=\VV_{\lambda_{\partial\Omega}+M}\varphi
=({\lambda_{\partial\Omega}+M})(\alpha\cdot N)\frac{\varphi_++\varphi_-}{2},
\end{equation*}
and therefore, using that ${\lambda_{\partial\Omega}+M}$ is real-valued, we get
\begin{equation}\label{l8eq12}
\varphi_-=\frac{2i+{\lambda_{\partial\Omega}+M}}{2i-{\lambda_{\partial\Omega}-M}}\varphi_+.
\end{equation}

The operator $U_\theta^{*}T_{\Lambda_{\lambda_{\partial\Omega}+M}}U_\theta$ defined on $U_\theta^{*}D\big(T_{\Lambda_{\lambda_{\partial\Omega}+M}}\big)$ is unitary equivalent to $T_{\Lambda_{\lambda_{\partial\Omega}+M}}$, thus it is self-adjoint. Since $\chi_{\Omega_-}(\alpha\cdot\nabla\theta)$ is a bounded and symmetric operator in $L^2(\mu)^4$ (because $\nabla\theta$ is continuous and compactly supported due to the fact that $\lambda\in\CC^1_c(\R^3)$), an application of the Kato-Rellich theorem (see \cite[Theorem X.12]{RS}, for example) finally shows that the operator $\big(T_{\Lambda_{\lambda_{\partial\Omega}+M}}\big)_\theta$ defined on $U_\theta^{*}D\big(T_{\Lambda_{\lambda_{\partial\Omega}+M}}\big)$ by \eqref{xyz3} is also self-adjoint. 

To conclude the proof of the theorem, it only remains to show that \eqref{xyz2} holds. Let $\psi\in U_\theta^{*}D\big(T_{\Lambda_{\lambda_{\partial\Omega}+M}}\big)$, so
$\psi=e^{i\theta\chi_{\Omega_-}}\varphi$ for some $\varphi\in D\big(T_{\Lambda_{\lambda_{\partial\Omega}+M}}\big)$. First, note that
\begin{equation}\label{l8eq13}
U_\theta^{*}T_{\Lambda_{\lambda_{\partial\Omega}+M}}U_\theta\psi
=U_\theta^{*}T_{\Lambda_{\lambda_{\partial\Omega}+M}}\varphi
=U_\theta^{*}(\chi_{\Omega_+}H\varphi+\chi_{\Omega_-}H\varphi),
\end{equation}
where we denoted by $\chi_{\Omega_+}H\varphi+\chi_{\Omega_-}H\varphi$ the absolutely continuous part of the distribution $H\varphi$ (it corresponds to $G$ when we write $\varphi=\Phi(G+g)$). Using the distributional equation \cite[(5.2)]{AMV2}, we see that
\begin{equation}\label{l8eq14}
\begin{split}
(H+\VV_{\lambda_{\partial\Omega}})\psi&=(H+\VV_{\lambda_{\partial\Omega}})e^{i\theta\chi_{\Omega_-}}\varphi
=H\big(e^{i\theta\chi_{\Omega_-}}\varphi\big)
+{\lambda_{\partial\Omega}}(\alpha\cdot N)\frac{\varphi_++e^{i\theta}\varphi_-}{2}\\
&=\chi_{\Omega_+}H\varphi
+\chi_{\Omega_-}(e^{i\theta}H(\varphi)
+(\alpha\cdot \nabla\theta) e^{i\theta}\varphi)\\
&\quad-i(\alpha\cdot N)(e^{i\theta}\varphi_--\varphi_+)
+\lambda_{\partial\Omega}(\alpha\cdot N)\frac{\varphi_++e^{i\theta}\varphi_-}{2}\\
&=U_\theta^{*}(\chi_{\Omega_+}H\varphi+\chi_{\Omega_-}H\varphi)
+\chi_{\Omega_-}(\alpha\cdot \nabla\theta) e^{i\theta}\varphi\\
&\quad+(\alpha\cdot N)\bigg\{\bigg(\frac{\lambda_{\partial\Omega}}{2}-i\bigg)e^{i\theta}\varphi_-+\bigg(\frac{\lambda_{\partial\Omega}}{2}+i\bigg)\varphi_+\bigg\}.
\end{split}
\end{equation}
From \eqref{l8eq12} and \eqref{xyz1} we obtain
\begin{equation*}
\begin{split}
\bigg(\frac{\lambda_{\partial\Omega}}{2}-i\bigg)e^{i\theta}\varphi_-+\bigg(\frac{\lambda_{\partial\Omega}}{2}+i\bigg)\varphi_+
=\bigg\{\bigg(\frac{\lambda_{\partial\Omega}}{2}-i\bigg)e^{i\theta}\frac{2i+{\lambda_{\partial\Omega}+M}}{2i-{\lambda_{\partial\Omega}-M}}+\bigg(\frac{\lambda_{\partial\Omega}}{2}+i\bigg)\bigg\}\varphi_+=0,
\end{split}
\end{equation*}
and then \eqref{l8eq14}, \eqref{l8eq13} and \eqref{xyz3} yield
\begin{equation*}
\begin{split}
(H+\VV_{\lambda_{\partial\Omega}})\psi&=U_\theta^{*}(\chi_{\Omega_+}H\varphi+\chi_{\Omega_-}H\varphi)
+\chi_{\Omega_-}(\alpha\cdot \nabla\theta)e^{i\theta}\varphi\\
&=U_\theta^{*}T_{\Lambda_{\lambda_{\partial\Omega}+M}}U_\theta\psi
+\chi_{\Omega_-}(\alpha\cdot \nabla\theta)\psi
=\big(T_{\Lambda_{\lambda_{\partial\Omega}+M}}\big)_\theta\psi
\end{split}
\end{equation*}
for all $\psi\in U_\theta^{*}D\big(T_{\Lambda_{\lambda_{\partial\Omega}+M}}\big)$. The proof of the theorem is complete.
\end{proof}

\end{document}